\documentclass{llncs}

\usepackage[algo2e, boxed, vlined]{algorithm2e}
\usepackage[utf8]{inputenc}
\usepackage[para]{footmisc}
\usepackage[table,xcdraw,dvipsnames]{xcolor}
\usepackage[numbers,sort&compress]{natbib}
\usepackage[english]{babel}

\usepackage{amsmath}
\usepackage{amssymb}
\usepackage{enumitem}
\usepackage{etoolbox}
\usepackage{float}
\usepackage{graphicx}
\usepackage{hyperref}
\usepackage{mathrsfs}
\usepackage{multirow}
\usepackage{subcaption}
\usepackage{url}
\usepackage{wrapfig}

\DeclareSymbolFont{letters}{OT1}{cmtt}{m}{n}
\DeclareMathSymbol{,}{\mathpunct}{operators}{"2C}

\hypersetup{
  colorlinks,
  linkcolor = {red!50!black},
  citecolor = {blue!50!black},
  urlcolor  = {blue!30!black}
}

\begin{document}

\title{A distance-based tool-set to track inconsistent urban structures through complex-networks}
\titlerunning{Tracking Inconsistent Urban Structures}

\author{Gabriel Spadon\inst{1} \and
Bruno B. Machado\inst{2} \and\\
Danilo M. Eler\inst{3} \and
Jose F. Rodrigues-Jr\inst{1}}
\authorrunning{G. Spadon et al.}

\institute{%
University of Sao Paulo, Sao Carlos -- SP, Brazil \and
Federal University of Mato Grosso do Sul, Ponta Pora -- MS, Brazil \and
Sao Paulo State University, Presidente Prudente -- SP, Brazil\\ \ \\
\email{spadon@usp.br},
\email{bruno.brandoli@ufms.br},\\
\email{daniloeler@fct.unesp.br},
\email{junio@icmc.usp.br}}

\maketitle

\begin{abstract}
Complex networks can be used for modeling street meshes and urban agglomerates.
With such a model, many aspects of a city can be investigated to promote a better quality of life to its citizens.
Along these lines, this paper proposes a set of distance-based pattern-discovery algorithmic instruments to improve urban structures modeled as complex networks, detecting nodes that lack access {\it from}/{\it to} points of interest in a given city.
Furthermore, we introduce a greedy algorithm that is able to recommend improvements to the structure of a city by suggesting where points of interest are to be placed.
We contribute to a thorough process to deal with complex networks, including mathematical modeling and algorithmic innovation.
The set of our contributions introduces a systematic manner to treat a recurrent problem of broad interest in cities.
\end{abstract}

\section{Introduction and Related Works}
\label{sect:introduction}

The synergy of real-world systems can be described as complex networks that exchange information through their entities' relationships.
Such networks can model complex systems from neuronal networks to subway systems~\cite{Boccaletti2006:StructureDynamics} and also, they can shape cities when linking the network topology with georeferenced data.

By analyzing the complex network of a city, it is possible to extract features that can describe urban problems, which are meaningful indicators for city planners~\cite{Porta2009:StreetCentrality}.
Such features can reveal, for instance, sites where social activities are more intense, regions where facilities should be placed, and neighborhoods that lack street access.
Particularly, these networks can expose distance-based inconsistencies, which is how we refer to nodes that lack efficient street access {\it from}/{\it to} others in the same network, possibly resulting in structural bottlenecks.

Along these lines, we identified a lack of methods to analyze and improve the structure, mobility, and street access of cities.
Consequently, in this paper, we contribute with a mathematical tool-set and algorithms to track distance-based inconsistencies by analyzing the complex-network topology of a city.
Our results have implications for the street access, supporting a finer street planning by enhancing mobility indicators and providing better city's structural assessment.

The core assumption of this study is that the network is supposed to provide streets that render the shortest distance between places.
In this regard, our tool-set uses two distance-functions to track nodes that do not provide shortest distance routes between them and other nodes that are of some interest.
Nodes that fail in providing minimum-length routes are considered to be inconsistent nodes, which are evidence of problems in the city structure.
Accordingly, in the face of an inconsistency, we raise two hypotheses:
{\bf (A)} the network lacks a more appropriate mesh; or, instead,
{\bf (B)} the city lacks its points of interest placed in better locations.
The first one indicates the need for new points of interest to distribute their load.
Contrarily, the second one indicates the need for relocating points of interest because the topology of the terrain cannot afford new streets.

A vast number of studies have been conducted to analyze inherent properties and behaviors found in cities.
For instance, multiple metrics have been adopted to explain their structural conditions~\cite{Masucci2013:UrbanGrowth}, their intense traffic of vehicles~\cite{Kaczynski2014:RoadTraffic}, and the emergence of collective behavior~\cite{Sopan2013:NationNeighbors}.
In other studies, the authors centered on the geometrical perspective of the network~\cite{Corcoran2015:InferringGeometry}, and on the elements positioning~\cite{Barthelemy2008:StreetPatterns, Zhong2014:DetectingDynamics}.
Furthermore, there are those who reviewed the role of the city elements~\cite{Crucitti2006:SpatialCentrality, Costa2010:TransportationSystems}, that addressed the support to the urban planning and design~\cite{Strano2012:GoverningEvolution,Spadon2017:UrbanInconsistencies}, and that improved the facility-location analysis and planning of street meshes~\cite{Li2016:SpatialOptimization}.
In addition to the ones that inspected the effectiveness of the underground systems~\cite{Costa2010:TransportationSystems} and the improvement of long-range connections~\cite{Viana2011:LongRange}, besides those who defined the concept of accessibility through complex networks and cities~\cite{Travencolo2008:Accessibility} and that tested the centrality of cities considering their space syntax~\cite{Crucitti2006:UrbanCentrality, Porta2006:Dual, Cardillo2006:PlanarGraphs}.

In this paper, we contribute with a tool-set that improve the analysis of cities by tracking inconsistent urban structures through complex networks.
This proposal follows by showing that the distance between two nodes can reveal ill-located points of interest and that such information can be used to make a city better distance-efficiently to their citizens.
To present our contributions, this paper is organized as follows:
Section~\ref{sect:mathematics} discusses our mathematical formulation and related algorithms;
Section~\ref{sect:results-discussions} discusses the results about the applicability of the proposed tool-set; and,
Section~\ref{sect:conclusion} presents the conclusions and final remarks.

\section{Mathematical Formulation and Algorithms}
\label{sect:mathematics}

\subsection{Preliminaries}

Along the text, we refer to a complex network as a distance-weighted directed graph $G=\{V, E\}$, which is composed by a set $V$ of $|V|$ nodes and a set $E$ of $|E|$ edges.
To model a city as a complex network, we considered streets as the edges and their crossings as the nodes, aiming to preserve their element's geometry.
An edge $e \in E$ is an ordered pair $\left<i, j \right>$, in which $i \in V$ is named {\it source} and $j \in V$ is named {\it target}, $i \neq j$.
Each node $i \in V$ has two properties $\{\mathcal{L}_{at}^{i}, \mathcal{L}_{on}^{i}\}$ that correspond to their coordinates --- $\mathcal{L}_{at}^{i}$ is the latitude and $\mathcal{L}_{on}^{i}$ the longitude.
Based on such coordinates, we conferred to the edges in $E$ a floating-point weight that refers to the great-circle (or {\it inline}) distance between their {\it source} and {\it target}.

Our mathematical tool-set tracks inconsistencies identified through distance functions to detect which element does not follow a pattern.
The pattern that we consider refers to the real-world distance between the nodes of the network, which in turn can provide insights about the locomotion through the city streets.
In this regard, we begin by tracking a set $P \subset V$ of points of interest;
the idea, then, is to determine two sets of nodes that surrounds a point of interest $p \in P$, which can reveal the city’s inconsistencies through applying algebraic operations.
We introduce these two sets as the {\it perimeter set of p} and the {\it network set of p}.

\subsection{Grouping nodes in the surroundings of points of interest}
\label{subs:sets}

The first set matches the closest nodes to a point of interest $p$ according to the great-circle (or {\it inline}) distance, which is referred as the {\it perimeter set of p}:
\begin{equation}
    {V}_{p}^{E}=\{v \in V | d_{vp}^{E}<d_{v\bar{p}}^{E}, \forall p \in P, \forall \bar{p} \in P, p \neq \bar{p}\}
    \label{eq:perimeter-set}
\end{equation}

\noindent
where $d_{ij}^{E}$ is the great-circle distance between $i$ and $j$ in the surface of Earth:
\begin{equation}
    d_{ij}^{E} = \mathcal{R} \times arcos \left(sin(\mathcal{L}_{at}^{i}) sin(\mathcal{L}_{at}^{j}) + cos(\mathcal{L}_{at}^{i}) cos(\mathcal{L}_{at}^{j}) cos(\bigtriangleup_{\mathcal{L}_{on}^{ij}})\right)
    \label{eq:great-circle}
\end{equation}

\noindent
where $\mathcal{L}_{at}^{i}$ and $\mathcal{L}_{at}^{j}$ are the latitudes, $\bigtriangleup_{\mathcal{L}_{on}^{ij}}$ is the difference between the longitudes $\mathcal{L}_{on}^{i}$ and $\mathcal{L}_{on}^{j}$, both of nodes $i$ and $j$.
Also, $\mathcal{R}$ is the radius of Earth (6,378 km), and all values are represented in radians.
Given a graph $G = \{V, E\}$ and a set of points of interest $P$, a node $v \in V$ pertains to the perimeter set of only one $p \in P$.

\begin{center}
    $\therefore$ $V_p^{E}$ and $V_{\bar{p}}^{E}$ --- $\forall p \in P, \forall \bar{p} \in P, p \neq \bar{p}$ --- are mutually disjoint.
\end{center}

The second set corresponds to the {\it network set of p}, which is made of nodes closest to a point of interest $p$ according to the length of their shortest paths:
\begin{equation}
    V_p^{N}=\{v \in V | d_{vp}^{N}<d_{v\bar{p}}^{N}, \forall p \in P, \forall \bar{p} \in P, p \neq \bar{p}\}
    \label{eq:network-set}
\end{equation}

\noindent
where $d_{ij}^{N}$ is the length of the shortest directed path ($spath_{ij}$) between $i$ and $j$ --- {\it i.e.} the sum of weights of all the edges in a minimum-length path, as follows:
\begin{equation}
    minimum\_length(spath_{ij}) = \sum_{e \in spath_{ij}}^{} weight(e)
    \label{eq:shortest-path}
\end{equation}

Recall that, the edge weight is given by the straight-line distance between their nodes using the great-circle distance (see Equation~\ref{eq:great-circle}).
We refer to the shortest path length as {\it network distance}, in the sense that one must necessarily (in the best case) move across this path to go from the source node to the target node.
Notice that any node $v \in V$ is network-closest to one and only one $p \in P$.

\begin{center}
    $\therefore$ $V_p^{N}$ and $V_{\bar{p}}^{N}$ --- $\forall p \in P, \forall \bar{p} \in P, p \neq \bar{p}$ --- are mutually disjoint.
\end{center}

In cases where the complex network is directed, the {\it network-distance TO a point of interest} is not necessarily the same as the {\it network-distance FROM a point of interest}, which may result in different network sets for the same $p$.
This detail is addressed in the following section, where we define the network set from a point of interest to the nodes in $V$ by mean of the {\it reversed network-set of p}:
\begin{equation}
    {\bar{V}}_p^{N}=\{v \in V | d_{pv}^{N}<d_{\bar{p}v}^{N}, \forall p \in P, \forall \bar{p} \in P, p \neq \bar{p}\}
    \label{eq:reversed-set}
\end{equation}

\subsection{Compartmentalizing inconsistencies for directed networks}
\label{subs:inconsistency-compartmentalization}

Consider different public services of a city as points of interest; such services may have different ways to assist the population, but all of them must require locomotion as a condition for assistance.
For example, in the case of doctors' clinics, it is desired that patients get there efficiently.
In turn, police stations require that their police officers efficiently reach the house of the citizens.
In the case of schools, the daily routine demands an efficient back-and-forth transit to students.
Along with other services that can be fitted with this assumption.
Notice that, we are referring to efficient paths as the ones with minimum length.

In the first example, there is an implicit displacement from a node $v$ to a node $p$; in the second one, the displacement is from the node $p$ to the node $v$; and, in the third case, there is a bi-directional displacement between $v$ and $p$, in which $v$ is an ordinary node and $p$ is a specific point-of-interest.
Based on the network direction, those three cases led to the following definitions:

\begin{enumerate}
    \item {\bf Inward Inconsistency:} nodes that are inline-closest to a point of interest, but network-closest (from $v$ to $p$, as given by Equation~\ref{eq:network-set}) to a different one:
    \begin{equation}
        {\Psi}_p^I= V_p^E - V_p^N
        \label{eq:inward-inconsistency}
    \end{equation}
    
    \item {\bf Outward Inconsistency:} the same as the previous category, but in the opposite direction (from $p$ to $v$, as given by Equation~\ref{eq:reversed-set}), resulting in the set:
    \begin{equation}
        {\Psi}_p^O= V_p^E - \bar{V}_p^N
        \label{eq:outward-inconsistency}
    \end{equation}
    
    \item {\bf Absolute Inconsistency:} nodes that are, simultaneously, considered to be inward and outward inconsistencies --- {\it i.e.} nodes in the sets' intersection:
    \begin{equation}
        {\Psi}_p^A={\Psi}_p^I \cap {\Psi}_p^O
        \label{eq:absolute-inconsistency}
    \end{equation}
\end{enumerate}

As mentioned, these categories rely on the direction of the network.
In cases where there is no direction, there will be no minimum-length divergence between paths of a round trip, but yet the inconsistencies can be tracked by calculating the difference between the perimeter set $V_p^E $ and the network set $V_p^N$ of $p$.
To provide further discussion, hereinafter we are considering just directed networks.

\subsection{Tracking distance-based inconsistencies}
\label{eq:tracking-inconsistencies}

In this section, we discuss Algorithm~\ref{alg:inconsistency-tracker} that joins the concepts that we previously introduced.
The aim of such algorithm is to track distance-based inconsistencies in distance-weighted directed networks by using a set $P$ of $|P|$ of points of interest.
Notice that, despite the definition of inconsistency is segmented into three types (see Section~\ref{subs:inconsistency-compartmentalization}), the algorithm considers a single inconsistency type at a time.

The algorithm starts by filling a set of empty sets, each one reserved to store the inconsistencies of a single point of interest (see lines~\ref{ln:alg-1:variables} to~\ref{ln:alg-1:initialization}).
Subsequently, we use $p^E$ and $p^N$ to store, respectively, the inline-closest and network-closest points of interest to a node $v \in V$ (see lines~\ref{ln:alg-1:inline-variable} and~\ref{ln:alg-1:network-variable}).
We used the external functions {\tt inline\_closest} and {\tt network\_closest} (see lines~\ref{ln:alg-1:inline-closest} and~\ref{ln:alg-1:network-closest}) to extract the closest point of interest to the node $v$; they implement, respectively, Equation~\ref{eq:great-circle} and~\ref{eq:shortest-path}.
Following, we perform a test to check whether a node is an inconsistency or not; thus, if the inline-closest point $p^E$ and the network-closest point $p^N$ are not the same (see line~\ref{ln:alg-1:inconsistency-test}) then $v$ is an inconsistency of $p^E$ (see line~\ref{ln:alg-1:inconsistency-set}).
Finally, a set of the inconsistencies of $|P|$ points of interest is returned as the result (see line~\ref{ln:alg-1:return-clause}).

\begin{algorithm2e}[!htb]
    \DontPrintSemicolon
    \LinesNumbered
    \KwData{$G = \{V, E\}$, $P \subset V$, and $c \in \{I, O, A\}$ --- $c$ is used to indicate the direction}
    \KwResult{$\left\{ {\Psi}_p^{c}, \forall p \in P \right\}$  --- a set of inconsistencies for all points of interest $p \in P$}

    ${\Psi}^{c} \leftarrow \emptyset$\;\label{ln:alg-1:variables}

    \For{{\bf each} $\ p \in P $}{
        $ {\Psi}^{c}_{p} \leftarrow \emptyset $\tcp*[l]{notice that ${\Psi}^{c}_{p} \subset {\Psi}^{c}, \forall p \in P$, therefore $|{\Psi}^{c}| = |P|$}\label{ln:alg-1:initialization}
    }

    \For{{\bf each} $\ v \in V $}{
    
        $ p^E \leftarrow \emptyset $\;\label{ln:alg-1:inline-variable}
        $ p^N \leftarrow \emptyset $\;\label{ln:alg-1:network-variable}

        \For{{\bf each} $\ \bar{p} \in P $}{
            $ p^E \leftarrow {\tt inline\_closest}(v, \left<p^E, \bar{p}\right>) $\;\label{ln:alg-1:inline-closest}
            $ p^N \leftarrow {\tt network\_closest}(v, \left<p^N, \bar{p}\right>, c) $\;\label{ln:alg-1:network-closest}
        }
        
        \If{$ p^E \neq \emptyset $ {\bf and} $p^N \neq \emptyset $}{
            \If{$ \left\{p^E\right\} - \left\{p^N\right\} \neq \emptyset $}{\label{ln:alg-1:inconsistency-test}
                $ {\Psi}_{p^E}^{c} \leftarrow {\Psi}_{p^E}^{c} \cup \{ v \} $\tcp*[l]{$v$ should be closer to $p^E$ than to $p^N$}
                \label{ln:alg-1:inconsistency-set}}}}
    \KwRet{${\Psi}^{c}$}\label{ln:alg-1:return-clause}
    \caption{An algorithm to track distance-based inconsistencies in cities modeled as networks.
    We use $p^E$ and $p^N$ to refer to the closest points of interest to a node $v$ considering, respectively, the inline and the network distances; other methods are related to the ones of Section~\ref{subs:sets}.}
    \label{alg:inconsistency-tracker}
\end{algorithm2e}

Given a graph $G = \{V, E\}$, a set $P$ of $|P|$ points of interest, and an inconsistent node $i$; such node is known to be an inconsistency to one and only one $p \in P$.

\begin{center}
    $\therefore$ ${\Psi}^{c}_{p}$ and ${\Psi}^{c}_{\bar{p}}$ --- $\forall p \in P, \forall \bar{p} \in P, p \neq \bar{p}$, $c \in \{I, O, A\}$ --- are mutually disjoint.
\end{center}

\noindent
Consequently, it is possible to derive two other sets from a point of interest $p$: {\bf (i)} the inconsistency set ${\Psi}^{c}_{p}$; and {\bf (ii)} the set of consistent nodes $\bar{{\Psi}}^{c}_{p} = V^{E}_{p} - {\Psi}^{c}_{p}$, such that $\bar{{\Psi}}^{c}_{p} \cap {\Psi}^{c}_{p} = \emptyset$.
The consistent nodes are fundamental to the process of suggesting locations to points of interest because they provide a smaller average distance to the nodes in their perimeter, different than an inconsistent node.

\subsection{Reducing distance-based inconsistencies}
\label{subs:reducing-inconsistencies}

In this section, we introduce Algorithm~\ref{alg:inconsistency-reducer}, which was designed to suggest changes in the location of points of interest to improve their access through the streets of a city.
The task of finding a perfect location for a point of interest might demand the test of all possibilities through an exhaustive search.
Consequently, our algorithm has a greedy approach that uses information about centrality metrics to guide the placement of a point of interest.
Centrality is not only an adequate technique to quantify the importance of a node but also it is capable to indicate central locations that are equally accessible to all nodes of a network.

Along these lines, we decided to adopt {\bf\textit{Straightness Centrality}}~\cite{Crucitti2006:SpatialCentrality} as the centrality metric of Algorithm~\ref{alg:inconsistency-reducer} because it analyzes the nodes of a network by joining both inline and network distances.
It is noteworthy that any distance-based centrality metric could be employed, as well as multiple metrics together; however, different metrics tend to provide dubious or bad choices for a relocation.

\begin{algorithm2e}[!b]
    \DontPrintSemicolon
    \LinesNumbered
    \KwData{$G = \{V, E\}$, $P \subset V$, and $c \in \{I, O, A\}$ --- $c$ is used to indicate the direction}
    \KwResult{{\it $R$} --- set of suggested positions for points of interest}
    
    $R \leftarrow \emptyset $
    $\bar{P} \leftarrow \emptyset$
    $\tilde{P} \leftarrow \emptyset$\;\label{ln:alg-2:variables}
    
    ${\Psi}^{c} \leftarrow {\tt algorithm\_1}(G, P, c)$\;\label{ln:alg-2:tracking-inconsistencies}
    ${\Phi}^{c} \leftarrow {\Psi}^{c}$\tcp*[l]{copy of the original set}

    \While{$|P| - |\bar{P}| > 0$ {\bf and} $\left( \sum_{i=1}^{|P|} |{\Psi}^{c}_{i}| \geq \sum_{i=1}^{|P|} |{\Phi}^{c}_{i}| \right)$}{\label{ln:alg-2:loop-until}
    
        $old_p \leftarrow \emptyset$\;
        $new_p \leftarrow \emptyset$\;
        
        \For{{\bf each} $\ p \in (P - \bar{P}) $}{\label{ln:alg-2:greedy-search}
            $G_{p}^{E} \leftarrow G\left(  V^{E}_{p} - {\Psi}^{c}_{p} \right)$\tcp*[l]{induced subgraph of consistent nodes}\label{ln:alg-2:induced-graph}
            $\bar{p} \leftarrow {\tt extract\_central}(G_{p}^{E})$\;\label{ln:alg-2:centrality}
            $\mathbb{P} \leftarrow \left( \left( \left( P - \bar{P} \right) \cup \tilde{P} \right) - \{p\} \right) \cup \{ \bar{p} \}$\;\label{ln:alg-2:temporary-relocation}
            ${\Omega^{c}} \leftarrow {\tt algorithm\_1}\left(G, \mathbb{P}, c \right)$\;\label{ln:alg-2:update-option}
            
            \If{$\left( \sum_{i=1}^{|P|} |{\Phi}^{c}_{i}| > \sum_{i=1}^{|P|} |{\Omega}^{c}_{i}| \right)$}{\label{ln:alg-2:check-minimum}
                ${{\Phi}^{c}} \leftarrow {{\Omega}^{c}}$\tcp*[l]{new lowest number of inconsistencies}\label{ln:alg-2:relocation-start}
                $old_p \leftarrow p$\;
                $new_p \leftarrow \bar{p}$\;
            }\label{ln:alg-2:relocation-ends}
        }
        
        \eIf{$old_p \neq \emptyset$ {\bf and} $new_p \neq \emptyset$}{\label{ln:alg-2:integrity-test}
            $R \leftarrow R \cup \{ old_p, new_p \}$\tcp*[l]{$old_p$ was moved to $new_p$}\label{ln:alg-2:do-relocate}
            $\bar{P} \leftarrow \bar{P} \cup \{ old_p \}$\tcp*[l]{old location}\label{ln:alg-2:remove-point}
            $\tilde{P} \leftarrow \tilde{P} \cup \{ new_p \}$\tcp*[l]{new location}\label{ln:alg-2:insert-point}
        }{
            {\bf break}\tcp*[l]{there are no more enhancements to be made}\label{ln:alg-2:break-code}
        }
    }
    \KwRet{$R$}\label{ln:alg-2:return-clause}
    \caption{An algorithm that uses the contributions of Algorithm~\ref{alg:inconsistency-tracker} to reduce distance-based inconsistencies of cities shaped as networks.}
    \label{alg:inconsistency-reducer}
\end{algorithm2e}

Our algorithm starts by initializing auxiliary variables (see line~\ref{ln:alg-2:variables}) and by tracking the inconsistent nodes in the original version of the network (see line~\ref{ln:alg-2:tracking-inconsistencies}).
In line~\ref{ln:alg-2:loop-until}, it starts looping until all points of interest have been replaced or until there are no more inconsistencies to be reduced from the original network.
After that, it tries to change one point of interest at a time (see line~\ref{ln:alg-2:greedy-search}).
The candidates to host a point of interest pertains to the induced subgraph $G^{E}_{p}$ of consistent nodes (see line~\ref{ln:alg-2:induced-graph}).
By using the induced subgraph the algorithm searches for the node that has the highest centrality value among all the other ones (see line~\ref{ln:alg-2:centrality}).

The algorithm continues by testing the highest central node as the new location to the point of interest; such that, it temporally replaces the node (see line~\ref{ln:alg-2:temporary-relocation}) and then it collects information about the inconsistencies of this network configuration (see line~\ref{ln:alg-2:update-option}).
Following, it tests whether the new configuration causes fewer inconsistencies then the previous one (see line~\ref{ln:alg-2:check-minimum}) before marking the node for relocation (see lines~\ref{ln:alg-2:relocation-start} to~\ref{ln:alg-2:relocation-ends}).
In a greedy fashion, it first selects the point of interest that by being replaced will lead to the highest elimination of inconsistencies.
After choosing the one to be replaced, we perform integrity tests, we mark the node as relocated, and then we remove it (see lines~\ref{ln:alg-2:integrity-test} to~\ref{ln:alg-2:insert-point}).

The algorithm ends when there are no more profitable changes (see line~\ref{ln:alg-2:break-code}).
It is noteworthy to mention that each point of interest can be moved only once; this is due to the greedy nature of the algorithm.
Otherwise, it would run until there are no more inconsistencies in the network at a prohibitive computational cost.
The output of the algorithm is a set $R$ of new locations (see line~\ref{ln:alg-2:return-clause}); each element $r \in R$ is an ordered pair $r=\left<old_p,\ new_p\right>$ that denotes the current ($old_p$) node where a point of interest is and a better node ($new_p$) for placing it.

Algorithm~\ref{alg:inconsistency-reducer} runs in $O(|V||P|^3)$ in the average case, where $|P|$ is the number of points of interest and $|V|$ is the number of nodes, $|P| \ll |V|$.
Besides that, the algorithm was designed to be straightly parallelized; and, moreover, in our tests, it took less than a minute to compute a whole city with 200,000 inhabitants.

\subsection*{Correctness of the algorithm formulation}
\label{subs:mathematical-proof}

In this section, we demonstrate that Algorithm~\ref{alg:inconsistency-reducer} is finite and it never increases the number of inconsistencies of a city, as required by the problem formulation.

\begin{theorem}
We hypothesize that Algorithm~\ref{alg:inconsistency-reducer} provides a set of central and consistent nodes that can replace specific points of interest in a city because replacing them will {\bf\textit{never increase the total number of inconsistencies}}.
\label{th:inconsistency-theorem}
\end{theorem}

\begin{proof}
Hereinafter, aiming to prove Theorem~\ref{th:inconsistency-theorem} by reduction to absurdity, we are supposing that the use of Algorithm~\ref{alg:inconsistency-reducer} can increase the number of inconsistencies.
Bearing in mind that the type of the inconsistency has no effect on such proof, we will follow by proving the algorithm using Inward Inconsistency (see Section~\ref{subs:inconsistency-compartmentalization}).

Consider the existence of a city mapped as a complex network $G = \{V, E\}$ and a set $P$ of $|P|$ points of interest located in this same city.
We start by finding the {\it perimeter set of p} ($V_{p}^{E}$) for each $p \in P$, which is given by Equation~\ref{eq:perimeter-set}.
Subsequently, we proceed with gathering the {\it network set of p} ($V_{p}^{N}$) that is defined by Equation~\ref{eq:network-set}.

Following, we detect a consistent node $\bar{p}$ that is the most central by an arbitrary centrality metric.
The most central node is the one that has the highest centrality when compared to the other nodes, potentially being a better place for positioning a point of interest in a city.
We follow by replacing $p$ by the most central node $\bar{p}$ in its perimeter.
Then, we calculate the updated perimeter ($V_{\bar{p}}^{E}$) and network ($V_{\bar{p}}^{N}$) sets, both of $\bar{p}$.
Notice that $p \neq \bar{p}$, thus $V_{p}^{E} \neq V_{\bar{p}}^{E}$ and $V_{p}^{N} \neq V_{\bar{p}}^{N}$.

At this point, there are two pairs of answers, one pair for $p$ and the other one for $\bar{p}$, as follows: $\left< V_{p}^{E}, V_{p}^{N} \right>$ and $\left< V_{\bar{p}}^{E}, V_{\bar{p}}^{N} \right>$.
The algorithm we proposed will replace $p$ by $\bar{p}$ following Equation~\ref{eq:proof-case}, which corresponds to a clause saying that the sets computed from $\bar{p}$ will be used just if they provide fewer inconsistencies than the original set; otherwise, it will keep the original one without making any changes.

\begin{equation}
{\Psi}_{p}^{I} =
    \begin{cases}
        V_{p}^{E} - V_{p}^{N}, \quad &\left|V_{p}^{E} - V_{p}^{N}\right| \leq \left|V_{\bar{p}}^{E} - V_{\bar{p}}^{N}\right| \\[1mm]
        V_{\bar{p}}^{E} - V_{\bar{p}}^{N}, \quad &otherwise
    \end{cases}
\label{eq:proof-case}
\end{equation}

The algorithm ceases when all the points of interest are changed at least once or when no change will result in inconsistency elimination (see Section~\ref{subs:reducing-inconsistencies}); as such, the algorithm is guaranteed to be finite.
Therefore, by reduction to absurdity, it is an \emph{absurdum} to suppose that the number of inconsistencies increases due to the use of Algorithm~\ref{alg:inconsistency-reducer} because the algorithm provides a set with less or equal inconsistencies than the original set --- as defined by Equation~\ref{eq:proof-result}.

\begin{equation}
    \therefore \left|{\Psi}_{p}^{I}\right| \leq \left|V_{p}^{E} - V_{p}^{N}\right|
    \label{eq:proof-result}
\end{equation}

\end{proof}

\section{Results and Discussions}
\label{sect:results-discussions}

The tool-set we proposed was validated over the Brazilian city of Sao Carlos.
Such city was instantiated as a complex network through a digital map from {\it OpenStreetMap}\footnote{\url{www.openstreetmap.org}}.
We considered streets as edges and their crossings as nodes; this way, we preserved the georeferenced attributes of the city that are necessary to the distance computation of our tool-set.
The resulting network is planar and it can be represented in two dimensions, in which edges intersect only at nodes.

\subsection{Assessing inconsistency recovery}
\label{sect:automated-recovery}

In this section, we analyze the inconsistent nodes found in the city of Sao Carlos regarding the location of hospitals, police stations, and public schools, which are our points of interest; such public services are known to be affected respectively by inward, outward, and absolute inconsistencies as described in Section~\ref{subs:inconsistency-compartmentalization}.
It is noteworthy that each set of points of interest are independent, as such, the inconsistencies of one set of points have no relationship with the ones of others.

The inconsistent nodes we tracked are in Table~\ref{tbl:inconsistency-information}, which suggest that their occurrence is connected to the number of points of interest.
In fact, they appear whenever different perimeters meet; as a consequence, there is no way to eradicate them without altering the network topology by changing the streets' direction or creating new streets.
In addition, more points of interest mean more boundaries, what tends to increase their number.
Hence, the challenge is to find locations to points of interest that reduce, rather than eradicate, inconsistencies.

We used Algorithm~\ref{alg:inconsistency-reducer} so to reduce the inconsistencies from Sao Carlos (see Table~\ref{tbl:inconsistency-information}).
The algorithm suggested relocating 6 hospitals, 2 police stations, and 9 public schools; such configuration, was able to reduce 160 inconsistencies from the hospitals (from 559 to 399), 123 inconsistencies from the police stations (from 342 to 219), and 179 inconsistencies from the public schools (from 663 to 484).
Notice that the inconsistencies of some points of interest raised in number from the original to the enhanced city, which is a setback of our approach.
However, as we have already proved, the total number of inconsistencies is always smaller.

\begin{table}[!htb]
    \centering
    \resizebox{.99\linewidth}{!}{
    \rowcolors{7}{gray!25}{white}
    \begin{tabular}{c|c|c|c|c|c|c|c|c|c|c|c|c|c}
        \hline
        \multirow{4}{*}{\rotatebox[origin=c]{90}{\bf $n^{th}$ POI}} & \multicolumn{6}{c|}{\bf Original City}
        &
        \multicolumn{6}{c|}{\bf Enhanced City} & \multirow{4}{*}{\rotatebox[origin=c]{90}{\bf $n^{th}$ POI}}
        \\ \cline{2-13}
        & \multicolumn{2}{c|}{\multirow{2}{*}{\it Hospitals}}
        & \multicolumn{2}{c|}{\multirow{2}{*}{\it Police Stations}}
        & \multicolumn{2}{c|}{\multirow{2}{*}{\it Schools}}
        & \multicolumn{2}{c|}{\multirow{2}{*}{\it Hospitals}}
        & \multicolumn{2}{c|}{\multirow{2}{*}{\it Police Stations}}
        & \multicolumn{2}{c|}{\multirow{2}{*}{\it Schools}}
        & \\
        & \multicolumn{2}{c|}{} & \multicolumn{2}{c|}{} & \multicolumn{2}{c|}{}
        & \multicolumn{2}{c|}{} & \multicolumn{2}{c|}{} & \multicolumn{2}{c|}{}
        & \\ \cline{2-13}
               & \#  & \%     & \#   & \%     & \#  & \%     & \#  & \%     & \#  & \%     & \#  & \%     &           \\ \hline
         01    & 013 & 02.3\% & 032  & 09.3\% & 015 & 02.2\% & 14  & 03.5\% & 30  & 13.7\% & 19  & 03.9\% & 01        \\ \hline
         02    & 002 & 00.3\% & 004  & 01.1\% & 077 & 11.6\% & 02  & 00.5\% & 48  & 21.9\% & 13  & 02.6\% & 02        \\ \hline
         03    & 012 & 02.1\% & 086  & 25.1\% & 043 & 06.4\% & 18  & 04.5\% & 96  & 43.8\% & 37  & 07.6\% & 03        \\ \hline
         04    & 019 & 03.4\% & 029  & 08.4\% & 071 & 10.7\% & 04  & 01.0\% & 32  & 14.6\% & 58  & 11.9\% & 04        \\ \hline
         05    & 030 & 05.3\% & 191  & 55.8\% & 114 & 17.1\% & 87  & 21.8\% & 13  & 05.9\% & 57  & 11.7\% & 05        \\ \hline
         06    & 049 & 08.7\% & ---  & ---    & 003 & 00.4\% & 51  & 12.7\% & --- & ---    & 01  & 00.2\% & 06        \\ \hline
         07    & 145 & 25.9\% & ---  & ---    & 008 & 01.2\% & 26  & 06.5\% & --- & ---    & 01  & 00.2\% & 07        \\ \hline
         08    & 039 & 06.9\% & ---  & ---    & 015 & 02.2\% & 22  & 05.5\% & --- & ---    & 18  & 03.7\% & 08        \\ \hline
         09    & 012 & 02.1\% & ---  & ---    & 078 & 11.7\% & 31  & 07.7\% & --- & ---    & 77  & 15.9\% & 09        \\ \hline
         10    & 043 & 07.6\% & ---  & ---    & 051 & 07.6\% & 63  & 15.7\% & --- & ---    & 48  & 09.9\% & 10        \\ \hline
         11    & 072 & 12.8\% & ---  & ---    & 038 & 05.7\% & 45  & 11.2\% & --- & ---    & 41  & 08.4\% & 11        \\ \hline
         12    & 095 & 16.9\% & ---  & ---    & 015 & 02.2\% & 17  & 04.2\% & --- & ---    & 11  & 02.2\% & 12        \\ \hline
         13    & 028 & 05.0\% & ---  & ---    & 056 & 08.4\% & 19  & 04.7\% & --- & ---    & 10  & 02.0\% & 13        \\ \hline
         14    & --- & ---    & ---  & ---    & 008 & 01.2\% & --- & ---    & --- & ---    & 16  & 03.3\% & 14        \\ \hline
         15    & --- & ---    & ---  & ---    & 060 & 09.0\% & --- & ---    & --- & ---    & 51  & 10.5\% & 15        \\ \hline
         16    & --- & ---    & ---  & ---    & 011 & 01.6\% & --- & ---    & --- & ---    & 26  & 05.3\% & 16        \\ \hline
          \rowcolor{gray!50}
     \bf Total & 559 & 100\%  & 342  & 100\%  & 663 & 100\%  & 399 & 100\%  & 219 & 100\%  & 484 & 100\%  & \bf Total \\ \hline
    \end{tabular}}
    \caption
    {Analysis of the inconsistencies of the city of Sao Carlos, in which we considered police stations, hospitals, and public schools as points of interest; we use \# to refer to the total number of inconsistencies and \% to their percentage.}
    \label{tbl:inconsistency-information}
\end{table}

\subsection{Supporting the designing of urban structures}

Our tool-set is not only to be used in the automatic recovery of inconsistencies, but also to assist human-made urban-planning decisions.
This is the case, for instance, when a specialist designs a city by having knowledge of the citizens' needs.
In this case, Algorithms~\ref{alg:inconsistency-tracker} and~\ref{alg:inconsistency-reducer} can aid the process by analyzing and recommending distance-efficiently locations that are feasible to points of interest.

This section introduces two hypothetical case studies that depict our tool-set in practice.
Both of them were conducted considering a subset of hospitals and public schools of the city of Sao Carlos (see Section~\ref{sect:automated-recovery}).
Nonetheless, our tool-set is extendable to any point of interest since it is equivalent to all of them.

Both case studies follow as in Figure~\ref{fig:example-case}, in which we start by finding a point of interest, next we try to solve the problem by ourselves, and then we use the algorithms to improve our results; all steps are guided under the light of the nodes' {\bf\textit{straightness centrality}}.
Furthermore, all case studies are represented by the induced subgraph of the point of interest being analyzed and, although we have illustrated the inconsistencies in Figure~\ref{fig:example-case}, in the case studies they are not visible because they do not provide visual information to the other images.

\begin{figure}[!htb]
    \centering
    \includegraphics[width=\linewidth]{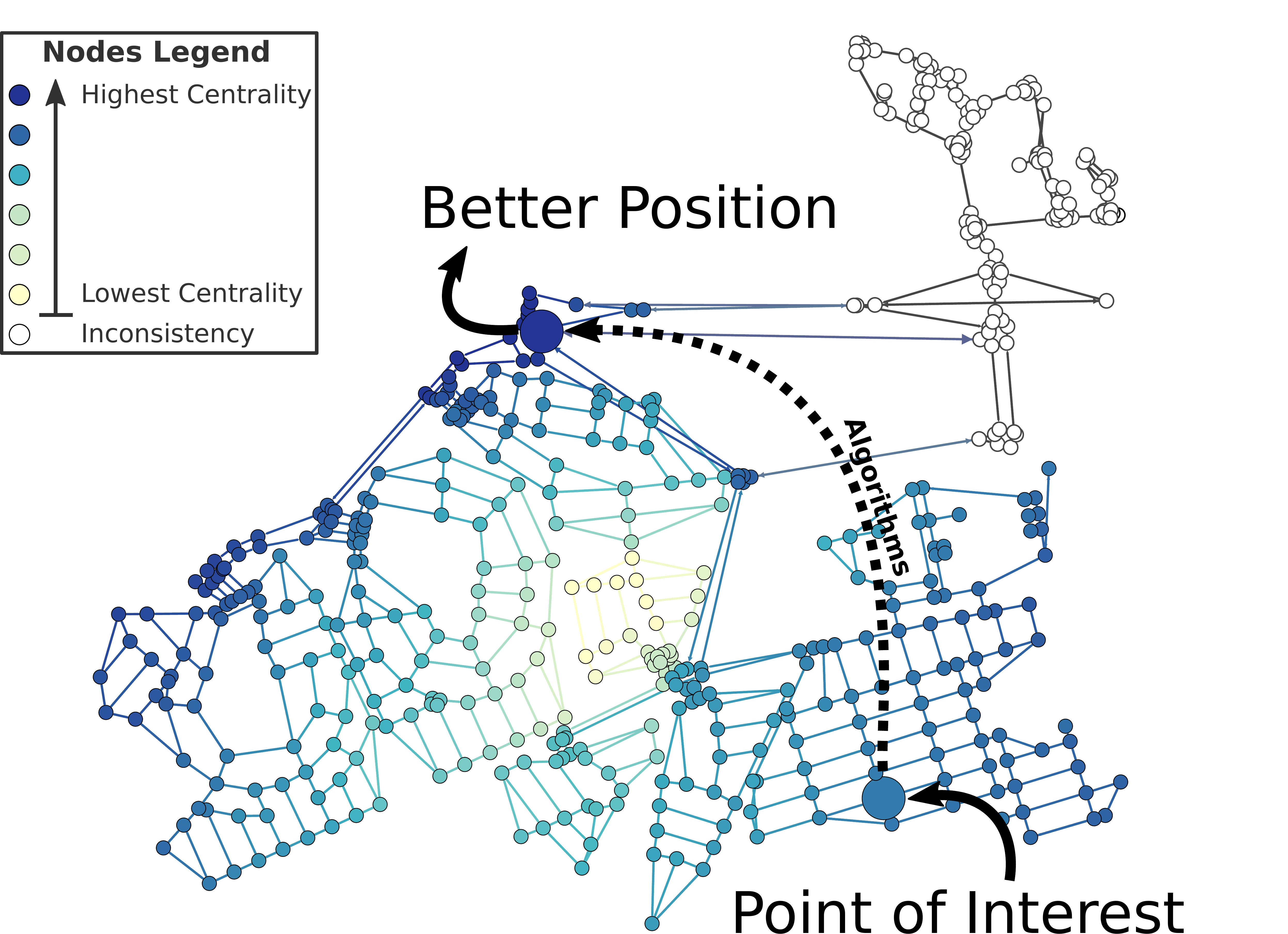}
    \caption{Illustration of the process of designing urban structures under the light of centrality metrics.
    This process starts by identifying nodes that are of interest, then it follows by tracking their inconsistencies, and it ends by suggesting new locations --- that reduce the number of inconsistencies --- to place these nodes.}
    \label{fig:example-case}
\end{figure}

\subsection*{Case Study 1: Creating a new hospital to reduce demand}

From the set of hospitals of the city of Sao Carlos, we identified one that, when compared to another hospital in the city, has excessive nodes in its perimeter (see Figure~\ref{fig:real-hospital}).
There is no specific explanation of the hospital's location and, for instance,  we can think that the city may have grown after the hospital has been built or the planners did not take the surroundings of the hospital into account.
One thing is for sure, an extensive area with an ill-positioned point of interest will deprive the street access of the nodes;
in this case, when points of interest are healthcare facilities, time-critical activities, as the transportation of patients in a critical state, can be jeopardized by lack of street access.
Hence, the problem becomes where to build a hospital and how to avoid inconsistencies.

\begin{figure}[!hbt]
    \centering
    \begin{subfigure}{.49\linewidth}
        \centering
        \includegraphics[width=\columnwidth]{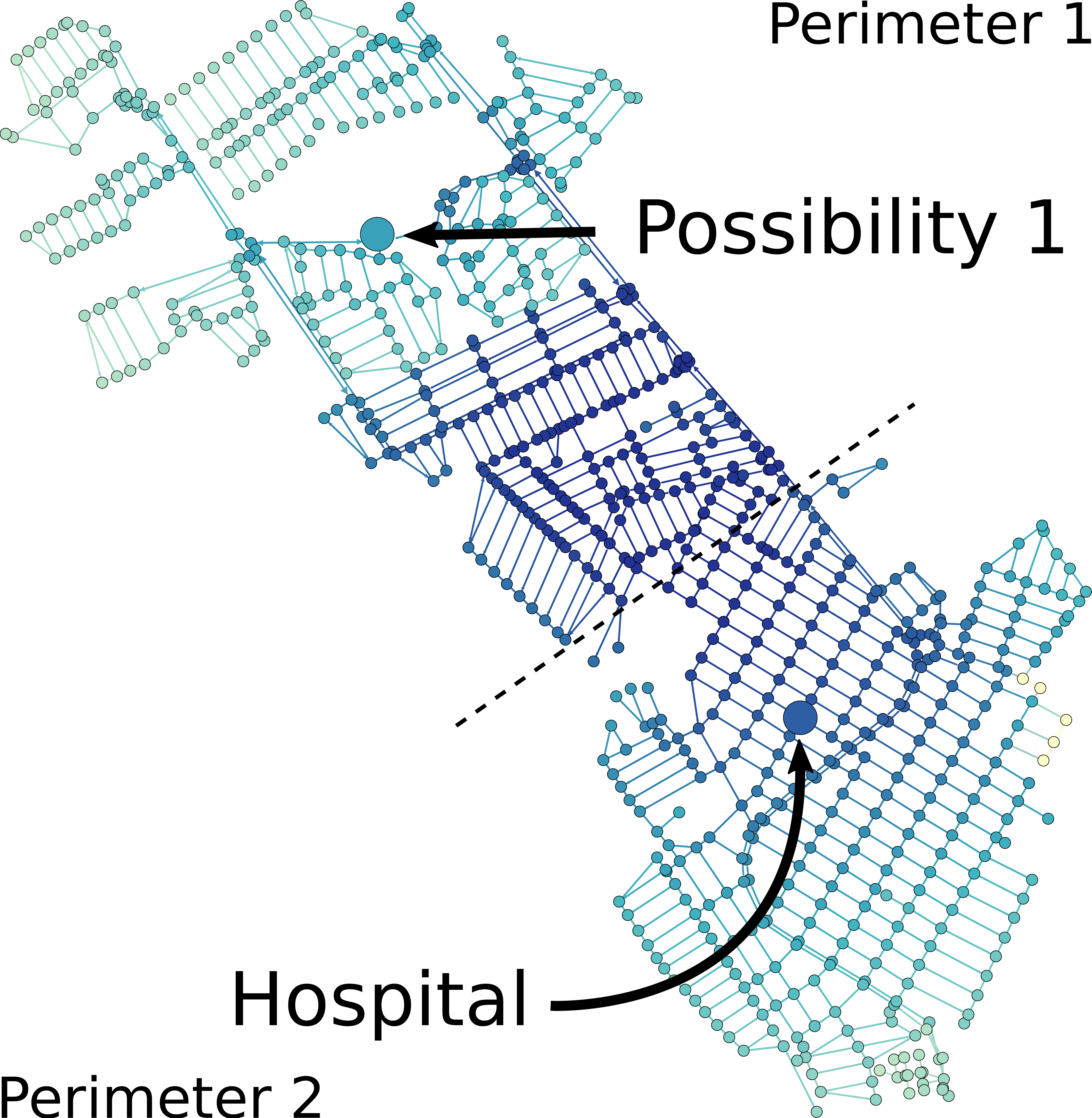}
        \caption{Original City}
        \label{fig:real-hospital}
    \end{subfigure}
    \hfill
    \begin{subfigure}{.49\linewidth}
        \centering
        \includegraphics[width=\columnwidth]{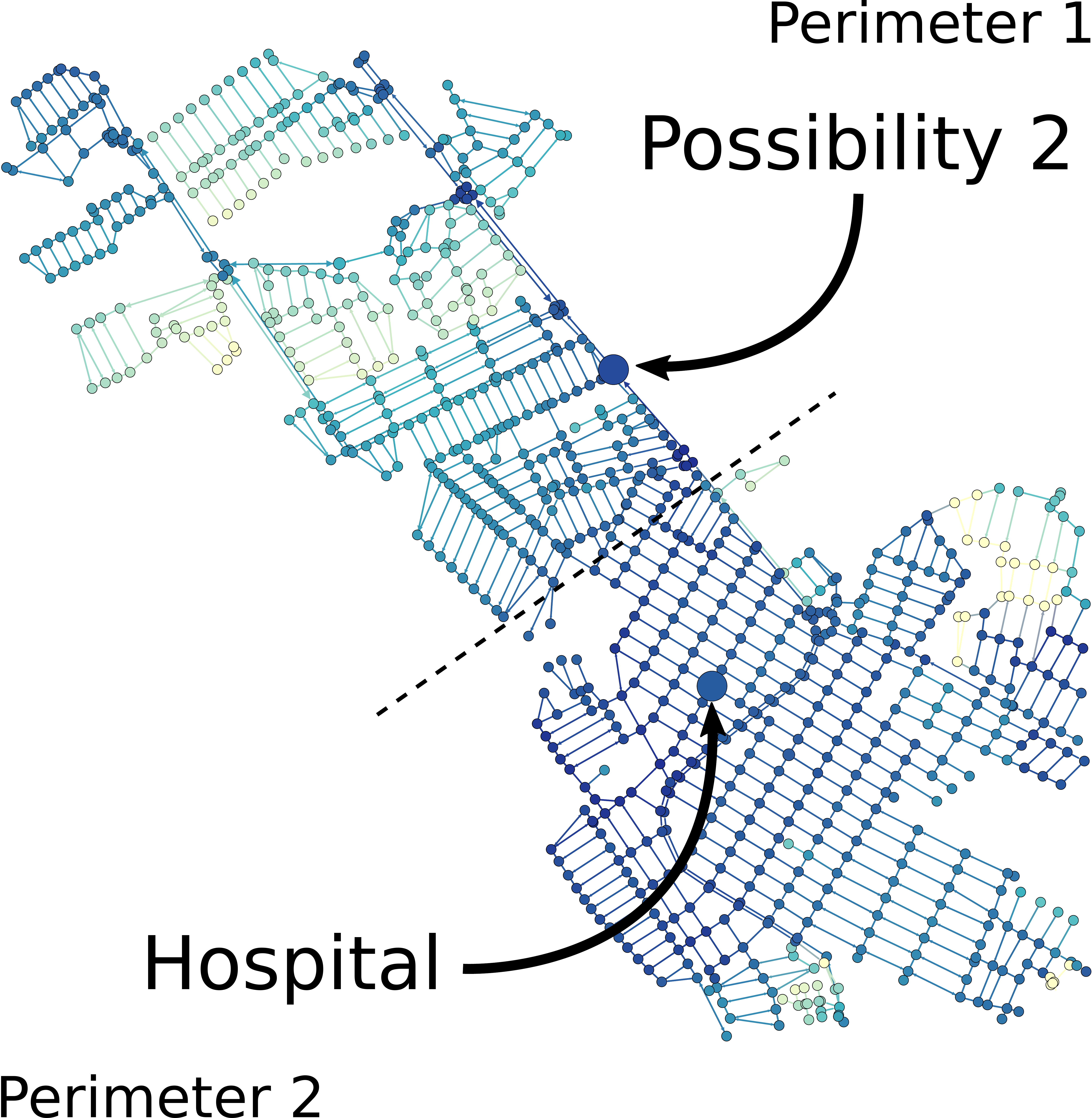}
        \caption{Enhanced City}
        \label{fig:result-hospital}
    \end{subfigure}
    \caption
    {Illustration of the assisted urban planning task from the first case study, in which the point of interest is a hospital and the color of the nodes denotes their centrality --- the darker, the higher.
    Figure~\ref{fig:real-hospital} shows a hospital's perimeter that is too large causing lack of access.
    We placed a new hospital in an eye-based central location in that same area to solve this issue.
    Afterwards, we used the algorithm to reduce inconsistencies, which suggested relocating the new hospital to a more central location that reduces the hospital's inconsistencies; as in Figure~\ref{fig:result-hospital}.}
\end{figure}

First, we tried to solve the problem manually by an eye-based analysis of a location that could provide equal nodes to the perimeters of both hospitals.
Figure~\ref{fig:real-hospital} shows a possible place to the new hospital as well as the resulting perimeter of both of them, which are defined by a line that cuts the image in half.
After that, we inserted the proposed location in the set of hospitals and we used Algorithm~\ref{alg:inconsistency-tracker} to track the inconsistencies of the resulting configuration.
Such configuration lead us to 615 inconsistencies, which is a bigger value than the original city.
Thus, we succeeded in building a hospital that splits the perimeter into two, but we failed in providing efficient access to both old and new hospital.

In a second approach, we analyzed the nodes' centrality together with a supporting visualization.
We colored the nodes by their centrality, what allowed us to notice that the selected location for the new hospital is a node with low centrality.
Then, we used Algorithm~\ref{alg:inconsistency-reducer} to suggest a better place for the new hospital while keeping the location of the old one.
Doing so, the city inconsistencies were reduced from 615 to 352 (see Figure~\ref{fig:result-hospital}), which positively reflected in the mobility of this area by distributing the demand between both hospitals.
Thus, creating a new hospital in a specific location was able to reduce almost half of the inconsistencies of the city without relocating the existing ones.

\subsection*{Case Study 2: Merging schools to centralize public resources}

In a similar fashion, we identified two public schools that are adjacent and support a short set of nodes.
In this case, the proximity of the schools (see Figure~\ref{fig:real-schools}) is a problem since none of them is used up to its capacity implying a waste of public resources.
In a first approach, by using Algorithm~\ref{alg:inconsistency-reducer} to relocate them, the number of inconsistencies was reduced from 663 to 635.

\begin{figure}[!hbt]
    \centering
    \begin{subfigure}{.49\linewidth}
        \centering
        \includegraphics[width=\columnwidth]{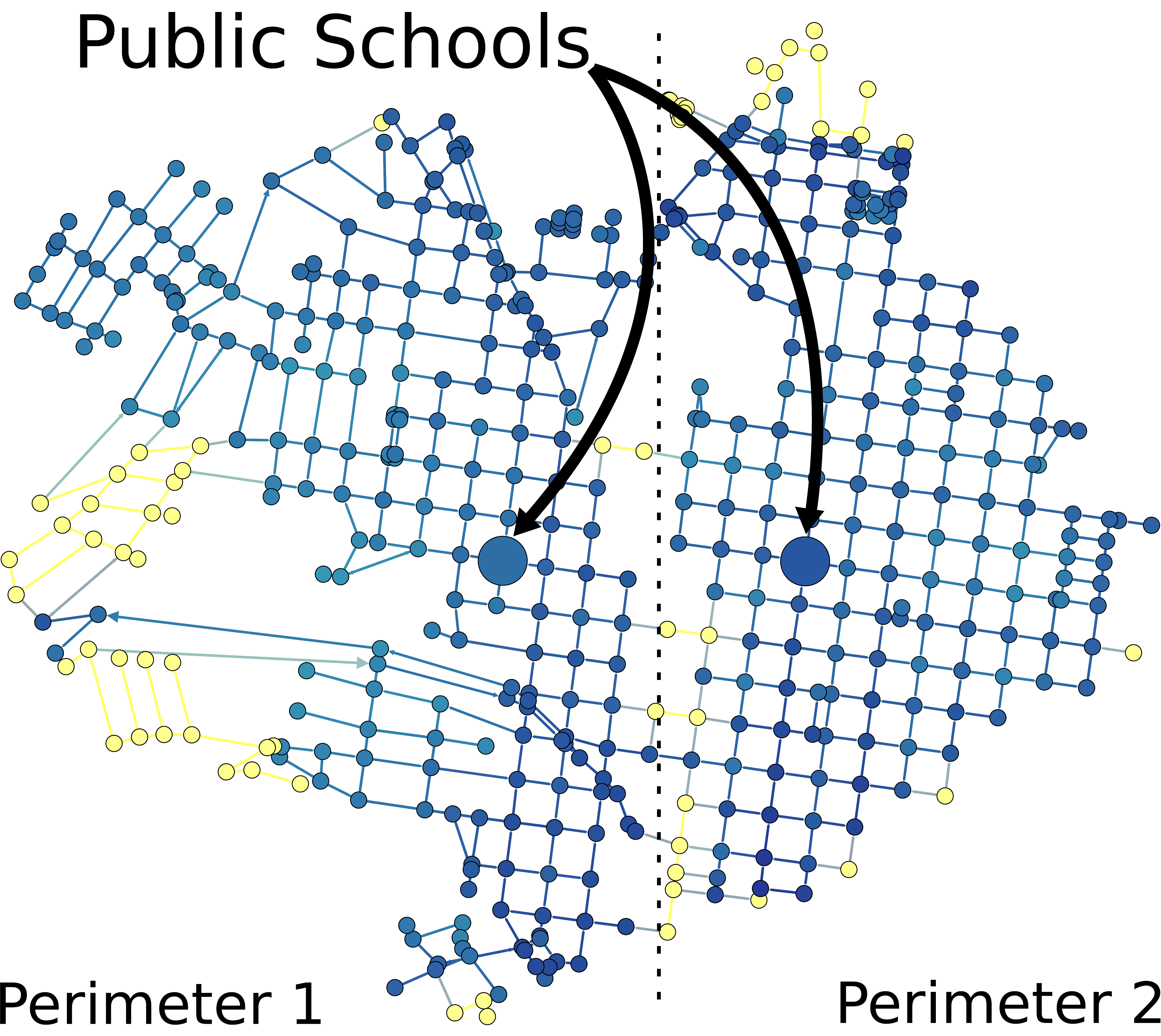}
        \caption{Original City}
        \label{fig:real-schools}
    \end{subfigure}
    \hfill
    \begin{subfigure}{.49\linewidth}
        \centering
        \includegraphics[width=\columnwidth]{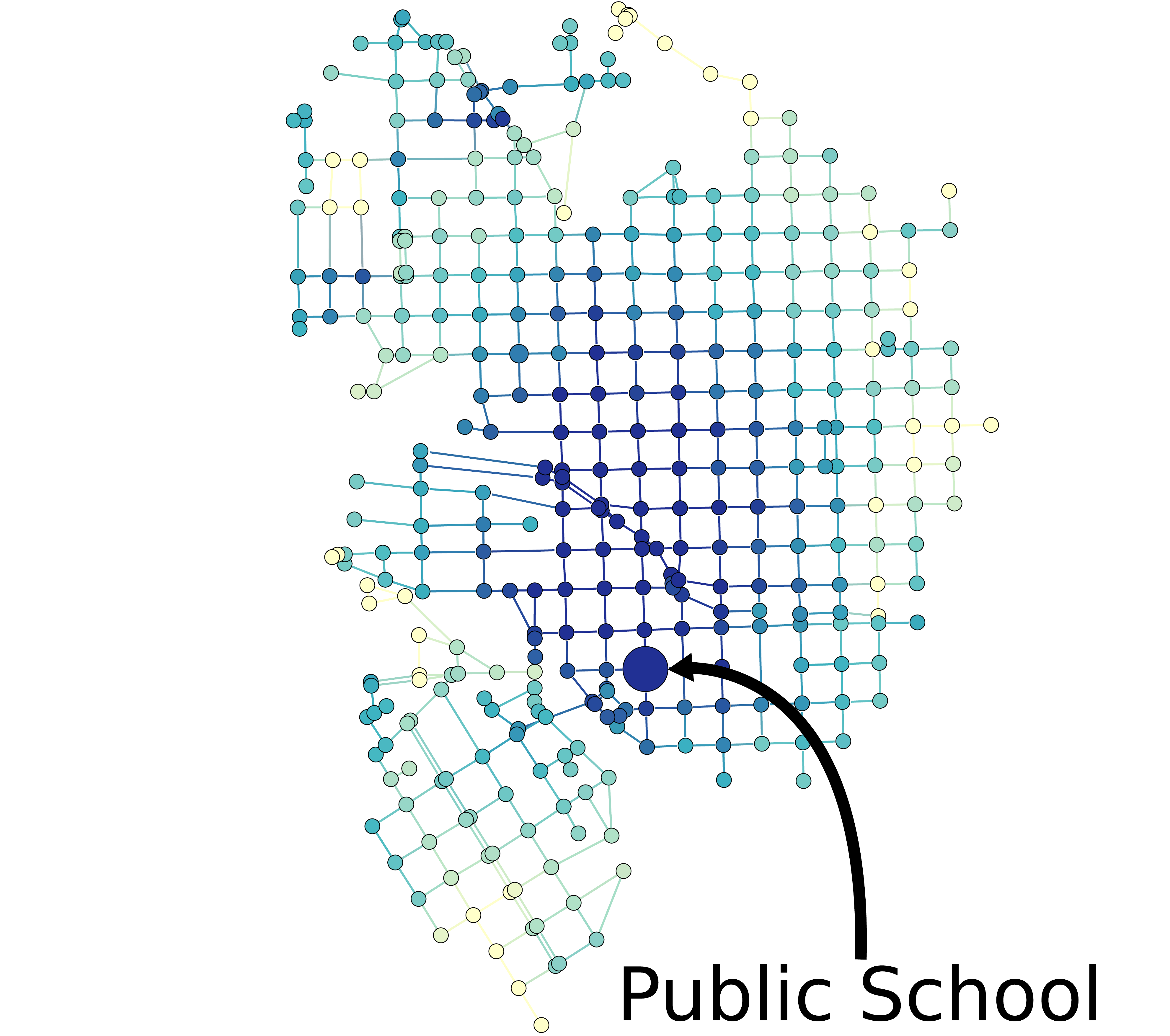}
        \caption{Enhanced City}
        \label{fig:result-school}
    \end{subfigure}
    \caption
    {Illustration of the assisted urban planning task from the second case study, in which the points of interest are public schools and the color of the nodes denotes their centrality --- the darker, the higher.
In this case study, we treated a problem related to the waste of resources that was caused by having two schools near each other; Figure~\ref{fig:real-schools} shows the problematic area, which is small, increasing the drawbacks related to access.
By replacing both schools with a single one we achieved a better coverage of nodes, as depicted in Figure~\ref{fig:result-school}.}
\end{figure}

Considering the size of the perimeter of both schools, we decided to remove one school to improve the utility of the one that remained.
By centralizing the schools in a single node, we can reduce inconsistencies because there will be fewer perimeters bordering each other; hence, the inconsistencies, located whenever two of them meet, will be naturally decreased.
To further enhance this process, we used the color-coded centrality metric to choose a candidate to be the new sole school.
Afterward, we used Algorithm~\ref{alg:inconsistency-reducer} to provide a better location (see Figure~\ref{fig:result-school}), which reduced the total number of inconsistencies from 635 to 445.

\subsection{Discussions on results generalization}

For a concise results presentation, we have assumed: {\bf (i)} that any displacement is through cities' streets; and, {\bf (ii)} a city with a uniform population distribution.
However, our tool-set holds for scenarios where these assumptions are not true.

We can use weights in accordance with the type of the displacement rather than using streets distance.
This is because our tool-set uses a general concept of weight and when providing additional information such weight can assume any quantitative value --- {\it i.e.} travel time, edge capacity, route cost, and so on.

About the population distribution, it is possible, for instance, to use a normal distribution peaked at the center of the city, multimodal distributions, or census data.
This information can aid in the analysis of urban agglomerations if it is used to assign values to sets of nodes corresponding to the population density of the area that they belong to.
Nevertheless, the set of inconsistencies would depend on the analysis of a specialist rather than being a self-explanatory result.

Also, despite being central to our problem formulation, the viability of redesigning a city is not suited for most cases. 
Furthermore, changing the topology of the network will alter the centrality of its elements, which will modify regions that attract vehicles and people.
Our tool-set is not only to be used in redesigning a city but also on the initial design when all possibilities are open.

Finally, our proposal has open problems that support further studies: 
{\bf (1)} the tool-set to track inconsistencies is categorical, then further algebra can aid in identifying the severity of a network inconsistency in a continuous, rather than binary, manner;
{\bf (2)} for simplicity's sake, we assumed the origin and destination of all paths as nodes of the network; such nodes are street intersections, which might not be real-world points of interest, requiring the addition of new nodes.

\section{Conclusion}
\label{sect:conclusion}

This paper was instantiated as a set of mathematical formalisms and algorithms to track and reduce distance-based inconsistencies improving access {\it to}/{\it from} points of interest in a city.
Beyond the mathematical formulation, we provided a proof of concept and case studies, all of which indicate that our tool-set is able to suggest better placements for points of interest at the same time that it improves the access to the majority of the nodes of a city by reducing its inconsistencies.

More specifically, our contributions are in the definition of a concept based on intrinsic problems to urban structures that are caused by the misallocation of points of interest in cities; also, in two algorithms that were devised to track and reduce inconsistent nodes in complex networks; and, finally, in a case study, in which we show how our tool-set and algorithms can aid planners and designers.

In summary, our methods were proved empirically and formally, granting potential for prompt contribution and for opening new research questions.
In addition, as a future work, we shall embrace link prediction methods for suggesting relocations in the network topology, {\it i.e.} proposing variations in the flow's direction, in the task of looking for a better topological setting for a city.

\section*{Acknowledgement}

We would like to thank the Brazilian agencies CNPq (167967/2017-7), FAPESP (2016/17078-0 and 2017/08376-0) and CAPES that fully supported this research.

\begin{small}

\end{small}

\end{document}